\documentclass{article}
\usepackage[a4paper, margin=1in]{geometry}
\usepackage[utf8]{inputenc}
\usepackage[hyphens]{url}
\usepackage{amsmath}
\usepackage{amsthm}
\usepackage{graphicx}
\usepackage{caption}
\usepackage{subcaption}
\usepackage{multicol}
\usepackage{multirow}
\usepackage{flushend}
\usepackage{authblk}
\usepackage{hyperref}
\usepackage{multirow}
\usepackage{lscape}
\usepackage{booktabs}
\usepackage{mathrsfs}

\newtheorem{theorem}{Theorem}

\title{The Variations of SIkJalpha Model for COVID-19 Forecasting and Scenario Projections}
\author{Ajitesh Srivastava}
\affil{University of Southern California\\ajiteshs@usc.edu}
\date{}
\usepackage{setspace}
\begin{document}

\maketitle

\begin{abstract}
    We proposed the SIkJalpha model at the beginning of the COVID-19 pandemic (early 2020). Since then, as the pandemic evolved, more complexities were added to capture crucial factors and variables that can assist with projecting desired future scenarios. Throughout the pandemic, multi-model collaborative efforts have been organized to predict short-term outcomes (cases, deaths, and hospitalizations) of COVID-19 and long-term scenario projections. We have been participating in five such efforts. This paper presents the evolution of the SIkJalpha model and its many versions that have been used to submit to these collaborative efforts since the beginning of the pandemic. Specifically, we show that the SIkJalpha model is an approximation of a class of epidemiological models. We demonstrate how the model can be used to incorporate various complexities, including under-reporting, multiple variants, waning of immunity, and contact rates, and to generate probabilistic outputs.
\end{abstract}

\section{Introduction}

Since the beginning of the COVID-19 pandemic, many models have been proposed to predict the trajectory of the outcomes including cases, deaths, and hospitalizations. To standardize the prediction tasks and communicate the results to stakeholders, several collaborative efforts (hubs) were initiated. These include US/CDC~\cite{lab_covid-19_2020}, Europe/ECDC~\cite{europe_forecast_hub} and Germany/Poland~\cite{bracher_assembling_nodate} forecast hubs, and the US and the European scenario modeling hubs~\cite{SMH,ESMH}. The forecast hubs ask for short-term predictions (up to 4 weeks ahead), while scenario modeling hubs ask for long-term projections (months to years ahead) under various scenarios. A collection of teams participate in these efforts, each one independently producing predictions and submitting regularly -- weekly for forecast hubs, and once in 1-2 months for scenario modeling hubs. 

During the early phase of the pandemic, we proposed the SIkJalpha model~\cite{srivastava_learning_2020,srivastava_fast_2020} for cases and deaths forecasts (see Section~\ref{sec:basic}). The approach has evolved with time, depending on dataset availability and new factors that we believed to have a significant impact. Further, for scenario modeling hubs, scenarios are defined on a certain decision or a variable of uncertainty. To account for the given scenario, that variable also needed to be included in the model. The modeling philosophy has been attempting to design the simplest model that can account for all such variables. 
We also avoid estimations where all parameters are estimated simultaneously. In prior work, we have shown that simultaneously estimating parameters of a simple non-linear version of our model can still lead to overfitting  
~\cite{srivastava_data-driven_2020}. Instead, we split the estimation problem into multiple sub-problems, most of which are solved by linear regressions. Some parameters are borrowed from the literature and some are provided by the specifying scenario. 

A key design feature of our approach is that all regions that a version of the model is applied to receive the same treatment. There is no region-specific manual tuning of the code. As an example, reporting rates differ significantly between the states in the US -- cases in Florida may be reported only once in two weeks, while some other states may report on all weekdays or certain days of the week. Yet, the same data pre-processing code applies to both. This is done to improve the generalizability of the approach, i.e., the same code can be applied to other regions in the world, as long as the data inputs the version of the model needs are available for those regions. A second key feature of the approach used for forecast hub submissions is that the forecasts are generated without any human interventions. Occasionally, hyperparameter-tuning is performed, when debugging or when new variables are to be included creating a new version of the model. The relevant hyper-parameters are presented in Table~\ref{tab:params}. The script is automated that generates all forecasts for submission. As a result, sometimes, data anomalies can cause unusual forecasts. However, we take the automated approach for scalability. The philosophy is that the approach should scale to a large number of regions if compute power and memory are available. Reliance on human assessment would limit the scalability. Further, the method should stand on its own without relying on the developer, and thus the open-source code can be used by anyone, including a non-expert. While these are desirable design features, it is crucial for the model forecasts to be accurate to be meaningful. Several works have evaluated our model demonstrating that it has performed well (see Section~\ref{sec:eval}).

\subsection{Contributions}

Our key contributions are the following. (i) We prove that the basic SIkJalpha model can approximate a class of epidemiological models (Theorem~\ref{thm:sir_special}). (ii) We present various versions of the SIkJalpha model that have been used over time over different forecast hubs (Section~\ref{sec:methodology}).(iii) We present the implementation details including data sources, parameter estimation methods, runtimes, and uncertainty quantification (Section~\ref{sec:uncertainty} and~\ref{sec:real-time}). Note that, over time, many modeling decisions were changed and new variables were included. The details of the impact of each modeling decision are beyond the scope of this paper. Instead, we focus on the breadth of modeling decisions. Evaluations of our models are publicly available on our dashboard as well as on others'  dashboards, articles, and reports (see Section~\ref{sec:eval}).

\subsection{Related Work}

During the COVID-19 pandemic, many models have been proposed for short-term and long-term forecasts. For the models that participated in the Forecast and Scenario Modeling Hubs, one way to categorize the model is the following.

\noindent \textbf{Mechanistic models} simulate the spread of COVID-19 through populations by explicitly representing the transmission dynamics of the virus. These models typically take into account factors such as the rate of contact between individuals, the probability of infection given contact, and the duration of the infectious period~\cite{suchoski2022gpu,linas2022projecting}.
These also include ``meta-population'' models that incorporate various demographics, spatial resolutions, and interactions between them~\cite{davis2021cryptic,lemaitre2021scenario}. Agent-based models may be considered a special case where each individual is modeled in the network of contacts~\cite{chen2021epidemiological}.

\noindent \textbf{Statistical models} use historical data to predict future COVID-19 outcomes. These include non-parametric spline-based models~\cite{wang2022nonparametric}, and weighted regression models that estimate growth rates~\cite{castro2021coffee}. A statistical model can be more sophisticated than mechanistic models but they are also more dependent on the quality and completeness of the data on which they are trained. 

\noindent \textbf{Machine Learning models} are also data-driven, but they can be distinguished from statistical models based on the fact they do not necessarily require an explicit model of the relationship between the variables. Most such models used some form of deep learning~\cite{zheng2021hierst,rodriguez2021deepcovid} to ingest data from multiple sources to predict future cases and mortality. 

\noindent \textbf{Ensemble models} have also been used to generate forecasts and projections based on multiple model outputs. Besides the ``Hubs'' that generate ensembles of submitted models, some submitting models also generate an ensemble of their own models. For instance, one submission combined auto-regressive models, Long short-term memory (LSTM) models, ensemble Kalman filter, and an SEIR model~\cite{adiga2021all}.

Scenario projections are more naturally performed by mechanistic models as the desired real-world variable can be directly incorporated. However, there have been Machine Learning-based models for scenario projections as well~\cite{UNCChierbin}.
Our approach is a combination of Mechanistic and Statistical category, as we model susceptibility as a mechanistic process. Given the susceptibility, the estimation of transmission parameters is done using a linear regression over new infections.

\section{Methodology}
\label{sec:methodology}
\subsection{The Basic SIkJalpha Model}
\label{sec:basic}
The key idea behind the SIkJalpha model~\cite{srivastava_fast_2020} is to approximate the disease dynamics as a discrete heterogeneous rate model. For instance, suppose the new infections are created by the following true dynamics:

\begin{equation}\label{eqn:main_infec}
\Delta I(t) = \frac{S(t)}{N}\mathcal{R}_0\sum_\tau \alpha(\tau)\Delta I(t-\tau) = \frac{S(t)}{N}\sum_\tau \alpha'(\tau)\Delta I(t-\tau)
\end{equation}
Here $S(t)$ is the number of susceptible individuals, $N$ is the population, and $\mathcal{R}_0$ is the reproduction number~\cite{dietz1993estimation} that represents the expected number of infections created by each infected individual in a fully susceptible population. The probability distribution of the serial interval, i.e., the delay between an infected individual infecting another, is given by $\alpha(\tau)$. As a result, someone who was infected at time $t-\tau$ on average will infect $\alpha'(\tau) = \mathcal{R}_0\alpha(\tau)$ individuals at time $t$.

Note that Equation~\ref{eqn:main_infec} represents the new infections rather than active infections which is common in variations of SIR models~\cite{allen1994some}. This choice has been made as new infections (daily/weekly) are more easily observed (adjusting for reporting rates and noise) compared to active infections. It can be shown that by proper choice of $\alpha(\tau)$ we can mimic the SIR model.

\begin{theorem}\label{thm:sir_special}
The SIR model is a special case of SIkJalpha model.
\end{theorem}
\begin{proof}
Suppose $\mathscr{I}(t)$ represents the number of active infections and $\mathscr{R}(t)$ represents the number of recovered at time $t$ in the SIR model. Then in the general model, total cumulative infections are the sum of active and recovered infections, i.e., $I(t) = \mathscr{I}(t) + \mathscr{R}(t)$. Therefore, new infections are the sum of new active infections and new recoveries.
\begin{equation}\label{eqn:active_cumu}
\Delta I(t) = \Delta\mathscr{I}(t) + \Delta \mathscr{R}(t) = \left( \beta \frac{S(t)}{N}\mathscr{I} - \Delta \mathscr{R}(t) \right) + \Delta \mathscr{R}(t) =  \beta \frac{S(t)}{N}\mathscr{I}(t).
\end{equation}
Suppose $\chi$ is the recovery rate in the SIR model such that $\Delta \mathscr{R}(t) = \chi \mathscr{I}(t)$. Also, the number of active infections at time $t$ is the sum of all new infections in the past that have not recovered by the time $t$. Therefore, under the SIR model:
\begin{align}
  \mathscr{I}(t) &= \sum_{t'= 0}^t \Delta I(t') P(\mbox{infected at $t'$ not recovered until $t$}) \nonumber \\
    &= \sum_{t'= 0}^t \Delta I(t') (1-\chi)^{t-t'} = \sum_{\tau = 0}^t \Delta I(t-\tau) (1-\chi)^{\tau}\,.
\end{align}
Putting this value in Equation~\ref{eqn:active_cumu}
\begin{equation}
    \Delta I(t) = \beta \frac{S(t)}{N}\mathscr{I}(t) = \beta \frac{S(t)}{N}\sum_{\tau = 0}^t \Delta I(t-\tau) (1-\chi)^{\tau} = \frac{S(t)}{N}\sum_{\tau = 0}^t \left(\beta (1-\chi)^{\tau}\right) \Delta I(t-\tau) \,.
\end{equation}
Setting $\alpha'(\tau) = \beta (1-\chi)^{\tau}$, this is equivalent to Equation~\ref{eqn:main_infec}.
\end{proof}

\begin{figure}
    \centering
    \includegraphics[width = 0.6\textwidth]{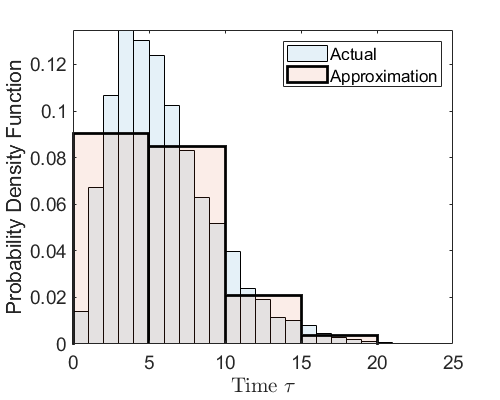}
    \caption{Demonstration of how the SIkJalpha approximates the distribution of the serial interval with $k=4$ and $J=5$.}
    \label{fig:serial_approx}
\end{figure}

We approximate the dependence on past infections by $k$ temporal bins of size $J$ as shown in Figure~\ref{fig:serial_approx}. Mathematically,
\begin{equation}\label{eqn:gen}
\Delta I(t) = \frac{S(t)}{N}\sum_\tau \beta'(\tau)\Delta I(t-\tau)\approx \frac{S(t)}{N}\sum_{i=1}^k \beta'_i \left(I\left(t-(i-1)J\right)-I(t-iJ)\right)\,.
\end{equation}
This enables the representation of arbitrary dependence with few learnable parameters. Further, in real-world setting, obtaining reliable daily data can be infeasible and may induce noise due to reporting patterns. Having $J>1$ has a smoothing effect. 

We note that only reported cases are observed, that can be represented by $\Delta R(t) = \sum_\tau \gamma(\tau)\Delta I(t-\tau)$, which when combined with Equation~\ref{eqn:gen} results in a further approximation:
\begin{equation}\label{eqn:report}
\Delta R(t)= \frac{S_t}{N} \sum_{i=1}^k \beta_i \left(R\left(t-(i-1)J\right)-R(t-iJ)\right)\,.
\end{equation}
Deaths (and hospitalizations) are considered to be outcomes of reported cases:
\begin{equation}\label{eqn:deaths}
    \Delta D_t = \sum_{i=1}^{k_D} \theta_i (R_{t-(i-1)J_D} -R_{t-iJ_D}).
\end{equation}
The assumption of independence of deaths and hospitalization enables independent learning of the parameters. Theoretically, the transition from cases to hospitalizations, and that from hospitalizations to deaths are conditionally independent, so they can be seen as a single transition from cases to deaths. However, if the data on cases is of poor quality, it will affect the prediction of deaths. In that scenario, it is better to estimate deaths from hospitalizations rather than cases in Equation~\ref{eqn:deaths}. The lag between the reported hospitalizations and reported deaths gives a portion of a clearer signal that is a better predictor of deaths.

The ``alpha'' in SIkJalpha represents the learning technique. We use a weighted linear regression to emphasize recently seen data to adapt for dynamically changing behavior. Suppose we have observed the reported cases for T time steps. Then the regression is performed by minimizing the weighted least square error
\begin{equation}
LSE = \sum_{t=1}^T \alpha^{T-t} \left( \Delta R_{obs}(t) - \Delta R_{fit}(t)\right)^2\,.
\end{equation}

Next, we present how various features were included and evolved over time. For a summary of changes over time across the Hubs, please see Table~\ref{tab:all_models}.

\subsection{True Infections}

Since not all infections are reported, we used several strategies over time to estimate the true number of infections. Our objective is to determine true new infections at a given time. 

\textbf{Estimated. } We assume that there exists a reporting probability $\gamma$ for each region, such that cumulative reported cases at time $t$ is given by $R(t) = \gamma I(t-t_0)$ for some constant $t_0$.
In our prior work~\cite{srivastava_data-driven_2020}, we showed that it is often not possible to reliably estimate under-reporting even in simple models. Whenever it is possible (certain periods and certain regions), we use the methods described in~\cite{srivastava_data-driven_2020}. Then we assume that similar reporting rates are observed in all other regions. We also assume that this reporting rate is constant in the past and will remain constant in the future. For short-term forecasts, in the early phases of the pandemic, the choice of this parameter did not impact the performance~\cite{srivastava2021epibench}.

\textbf{Seroprevalence. } Wherever available (for US states),  we use seroprevalence data~\cite{noauthor_nationwide_nodate} that provides an estimated total number of individuals who had infection-induced immunity at the frequency of once per 2-3 months. Each of these points provides an estimate of reporting probability $\gamma(t)$. We use linear interpolation to fill in the missing days within each of these 2-3 month periods. 
Finally, we take a weighted average to find the average reporting probability:
\begin{equation}
    \gamma = \frac{\sum_t \gamma(t) \Delta R(t)}{R(T)},
\end{equation}
where $T$ is the length of the time series of reported cases. The seroprevalence dataset also provides an interval for the percentage population with infection-induced immunity using which we get an interval $(\gamma_{low}, \gamma_{high})$ of reporting probability. 
Note that this assumes a constant reporting rate in the past, and we use the same rate to project in the future. While a constant reporting rate in the past may not be an accurate reflection of the truth, we consider it a reasonable modeling choice for forecasting and scenario projection as it still provides a good estimate of total infections.
Also, note that the number of individuals with infection-induced immunity is less than the number of total infections as one individual may have multiple infections. However, before the rise of the omicron variant, the number of reinfections was orders of magnitude smaller~\cite{NY_breakthrough} than uncertainty in true infections introduced due to the range $(\gamma_{low}, \gamma_{high})$. For the regions where seroprevalence data was not available we used similar range for $(\gamma_{low}, \gamma, \gamma_{high})$.

\textbf{Wastewater. } Due to the low frequency of updates in seroprevalence data and increase in reinfections, we switched to a wastewater dataset in early 2022 for US states~\cite{biobot}. Instead of a constant reporting rate assumption, we use the following model
\begin{align}
    \Delta I(t) &\propto C(t) = c_0 C(t)\, \\
    \gamma(t)  &= I(t)/R(t) \label{eqn:ww_past_reporting} \\
    \gamma(t) &= \gamma(T) \forall t > T\,, \label{eqn:ww_future_reporting}
\end{align}
where $T$ is the last observed time and $C(t)$ is the average effective concentration at time $t$. We also apply a cubic spline smoothing~\cite{reinsch1967smoothing} on $\gamma(t)$.
To estimate $c_0$, we assume that seroprevalence data between August 2020 ($t_1$) to July 2021 ($t_2$) provides a good estimate of incident infections. Then,
\begin{equation}
    c_0 = \frac{\sum_{t=t_1}^{t_2} \Delta I(t)}{\sum_{t=t_1}^{t_2} C(t)}\,.
\end{equation}
This enables estimation of $\Delta I(t)$ over the entire period. While this results in variable reporting probability in the past (Equation~\ref{eqn:ww_past_reporting}), for forecasting and projections, we still use a fixed reporting rate based on most recent rate (Equation~\ref{eqn:ww_future_reporting}).

\subsection{Variants}

During the emergence of the Alpha variant, we incorporated it into the model. Our general approach is to split the cases into multiple time series representing the cases caused by each variant.
While genomic sampling data is available that can provide proportions over time, a major challenge is that the sampling data is updated with a lag of 1-2 weeks~\cite{covariants}. Instead of using the variants proportions directly, we fit a model to identify the proportions as a function of time. This enables us to fill in for the days with missing data and removes the noise due to sampling. Specifically, we use the last 2-3 months of data to fit a logistic model, so that:
\begin{align}
    \log \frac{\Delta I^{(i)}(t)}{\Delta I^{(1)}(t)} = s_it + c_i \,,
    \sum_i \Delta I^{(i)}(t) = \Delta I(t)\,,
\end{align}
where $\Delta I^{(i)}$ represents the infections created by the variant $i$.
We may interpret $s_i$ as the advantage in the growth rate of variant $i$ over the base variant $1$. The model assumes a constant growth rate advantage within the short period under consideration.
In the long-term, the variants dynamics can significantly deviate from the logistic model. However, in the short term, it can be shown that the deviation is small. A detailed analysis of the errors due to this violation is beyond the scope of this paper. Various implementations were used over time:

\textbf{2-strain. } In this model, we only focus on two strains - two most dominant strains, or the current dominant strain, and one strain known to be rapidly increasing in prevalence. E.g., wild-type vs Alpha, and non-Omicron vs Omicron. A disadvantage of this model is that if there are more than two variants with significantly different transmission rates, then forcing them into two groups may lead to variable growth advantage deviating from the logistic model, even in the short term. Another drawback is the human intervention requirement -- two strains need to be manually selected and updated when a new highly transmissible variant appears.

\textbf{All PANGO. } In this model, we use all the variants that appear as different PANGO lineages~\cite{pango}. This approach does not require manual selection of variants and automatically includes new variants that appear as the dataset~\cite{genbank,gisaid} is updated.
However, this approach also has some drawbacks as the number of different lineages active in a 3-month window can be of the order of a few 100s. As a result, the computation time increases. Also, distributing the samples across so many lineages results in higher uncertainty due to the small number of samples per lineage.

\textbf{Selected. } In this model, we take the list of all PANGO lineages, and filter it to select those of interest. For example, during the Omicron wave, we extracted all lineages with names starting with ``BA''. All other lineages that appeared before omicron are grouped into one variant called ``pre-omicron''.

\subsection{Vaccines and Waning Immunity}
Initially, a simple \textit{all-or-nothing model} model was used, where the acquired immunity (through vaccines or prior infection) against variant $i$ provides full protection to $(1-\epsilon_A(i, j))$ fraction of the total population $N_j$ in age group $j$ at time $t$, and no protection to the rest. As a result, the susceptible population is given by 
\begin{equation}\label{eqn:vacc_AoN}
    S_{j}(t) = N_j - I_{j}(t) - \epsilon_A(i, j)V(t)\left(1 - I_{j}(t)/N_j\right)\,.
\end{equation}
This assumes that infection provides full immunity and those who are vaccinated but not infected are susceptible with probability $\epsilon_A(i, j)$.
The value of the parameter $\epsilon_A(i, j)$ was borrowed from the vaccine trials~\cite{oliver2020advisory,SMH}. For simplicity, we assumed $\epsilon_A(i, j) = \epsilon_A(j) \forall i $.
However, since July 2021, a \textit{time-waning model} was used in accordance with the discussions with the Scenario Modeling Hub. In this model, immunity is lost with time based on some function $w_{j}^{(i)}(\tau)$, where $\tau$ is the time since vaccination/infection. The expected susceptible population is given by
\begin{equation}\label{eqn:vacc_time}
    S_{j}(t) \approx N_j - I_{j}(t) - V(t)\left(1 - \frac{I_{j}(t)}{N_j}\right) + \left(\sum_\tau w_{j}(\tau)\left(\Delta V_j(t-\tau) + \Delta I_j(t-\tau)\right) \right)\,.
\end{equation}
From Aug 2021 - Oct 2021, the function $w_{j}^{(i)}$ was assumed to be a function representing the expected susceptibility based on the following model: (i) An individual transitions to a ``partially susceptible state at a time given by exponential distribution with mean $\lambda$. (ii) In this state, the individual has a residual protection $p_{inf}$ against infection and $p_{death}$ and $p_{hosp}$ against death and hospitalization, respectively. This means that they will only be susceptible with probability $1-p_{inf}$.
The parameters were determined by the scenarios in Scenario Modeling Hub. Since then, we have switched to using a gamma distribution to represent the transition into the partially susceptible state. The gamma distribution requires two parameters to be specified. These are obtained from two equations: 1) The mean/median of transition probability should match the specification of the scenarios. 2) The expected immunity after two months should match the vaccine trial results obtained at two months.

Note that both Equations~\ref{eqn:vacc_AoN} and~\ref{eqn:vacc_time}, are approximations in that they only support one breakthrough of acquired immunity and that total infections is roughly equal to the total individuals infected. With increasing reinfections, this approximation does not hold, and a more sophisticated model is used later (see Section~\ref{sec:all_state}).

\subsection{Vaccine Adoption}

During the early phase of vaccinations, from Jan 2021 to July 2021, a linear approximation was used for short-term forecasts (up to 4-weeks ahead). We use linear extrapolation determined by the last two weeks of new vaccinations to get the future new vaccinations. For longer-term projections, a sigmoid model was used:
\begin{equation}
    V(t) = \frac{a}{1 + \exp(-b(t-c))}\,
\end{equation}
where $b$ and $c$ are learnable parameters based on early observed data. The parameter $a$ is fixed based on vaccine adoption surveys.

When a few month of data was available, we modeled vaccine adoption as a contagious behavior:
\begin{align}\label{eqn:vacc-train}
    \Delta V(t) &= \frac{E_V(t) - V_E(t)}{N}\sum_{i=1}^k \beta_i \left(V(t-(i-1)J) -V(t-iJ)\right)
\end{align}
where $V(t)$ is the cumulative number of people vaccinated by time $t$, $E_V(t)$ is the number of individuals eligible for the vaccine at $t$, and $V_E(t) \leq V(t)$ is the number of individuals who are eligible at time $t$ but are already vaccinated. The above model may be applied to different age groups independently. For the first dose, everyone may be eligible in a given age group. We first predict the future-first doses. This determines the eligibility for the second dose, which in turn determines the eligibility for the booster. In the implementations so far, $V_E(t) = V(t)$ as the eligibility is cumulative - someone who is eligible at $t$ will remain eligible at $t' > t$.

The hyperparameters are set to $k=2, J=7$. The parameters $\beta_i$ are estimated using a weighted regression similar to the one described in Section~\ref{sec:basic} with a weight of 0.9.

To simulate target adoption scenarios, such as targeting a fraction $u$ of eligible population, a scaling function $\alpha(t)$ is introduced, after estimating $\beta_i$ in Equation~\ref{eqn:vacc-train}.  

\begin{align}
       \Delta V(t) &= \frac{uE_V(t) - V_E(t)}{N} \sum_{i=1}^k \alpha(t)\beta_i \left(V(t-(i-1)J) -V(t-iJ)\right).
\end{align}
We pick $\alpha(t) = \exp(\frac{H}{H-t})$, where $H$ is the time before which the target adoption is to be reached. The impact of this function is that it smoothly and rapidly increases the rate with time, thus making the term $\frac{uE_V(t) - V_E(t)}{N}$ decisive in the saturation of adoption.

\subsection{Age-groups}
\label{sec:contact}
When age-specific targets (cases, deaths, and hospitalizations) are available, we incorporate interactions between age-groups in a ``contact matrix''. We assume different age groups may have different susceptibility and there is a transmission rate that is scaled by the contact matrix $C=[c_{g,g'}]$, where $c_{g, g'}$ represents a relative measure of contact rates between age groups $g$ and $g'$. 
In other words, $c_{g, g'}$ represents the impact of age-group $g'$ on $g$. For $G$ age groups, this may introduce $G(G-1)/2$ parameters if we assume symmetric impacts.
To avoid overfitting, we reduce this to $G$ parameters, assuming that for some parameters $c_1, c_2, \dots, c_G$, $c_{g, g'} = c_g {c_g'}$.   
One way to interpret this is to assume that age-group $g$ independently participates in a randomly selected contact with probability $\propto c_g$. Therefore, probability of an individual in age group $g$ being in contact with one in $g'$ is proportional to $c_g c_{g'}$. This is similar to the formulation of the contact matrix in Ma et. al.~\cite{ma2021modeling}.

Since contact is a physical process, it is assumed to be the same for all variants. We rewrite Equation~\ref{eqn:report} to incorporate age-groups as:
\begin{equation}\label{eqn:report_final}
    \Delta R^{(i)}(t, g) = \frac{S^{(i)}(t, g)}{N(t, g)} \sum_{l=1}^k \sum_{g'=1}^{G} \beta_l^{(i)}c_g c_{g'} \left(R^{(i)}(t-(l-1)J, g') - R^{(i)}(t-lJ, g')\right)\,.
\end{equation}
To fit the model to the data, first, we fix the variant $i$ by choosing the one which has the most cumulative cases in the recent weeks (in our implementation, we used the last four weeks). Then, we set $c_G = 1$, for the last age group, and learn $\beta_l^{(i)}$ and $c_1, \dots, c_{G-1}$. Then, we fix the learned values of $c_g$ and learn $\beta_l^{(i)}\, \forall l, i$. The purpose of setting $c_G = 1$ is to avoid multiple equivalent solutions obtained by arbitrarily scaling up all $c_g$ and scaling down $\beta_l^{(i)}$.

As an example, the following matrices are obtained from this method based on data as of June 6, 2021 for California and Washington, respectively.
\begin{center}
\begin{tabular}{ccc}
$ \begin{bmatrix}
0.88 & 0.88 & 1.15 & 2.13 & 0.94 \\
0.88 & 0.89 & 1.16 & 2.15 & 0.94 \\
1.15 & 1.16 & 1.51 & 2.80 & 1.23 \\
2.13 & 2.15 & 2.80 & 5.19 & 2.28 \\
0.94 & 0.94 & 1.23 & 2.28 & 1.00 
\end{bmatrix}  $
     &
     \hspace{1cm}
     &  
     $ \begin{bmatrix}
 1.02 & 1.01 & 1.32 & 2.40 & 1.01 \\
 1.01 & 1.01 & 1.32 & 2.40 & 1.01 \\
 1.32 & 1.32 & 1.73 & 3.13 & 1.31 \\
 2.40 & 2.40 & 3.13 & 5.66 & 2.38 \\
 1.01 & 1.01 & 1.31 & 2.38 & 1.00 
\end{bmatrix}  $
\end{tabular}
\end{center}
The age groups considered here are 0-4, 5-11, 12-17, 18-65, 65+. The entry in the $i^{th}$ row and $j^{th}$ column represent the relative contact rate between age groups $i$ and $j$. The matrix is normalized so that the contact rate within the oldest age group is 1.
Figure~\ref{fig:age_fit} shows the fitted time-series by age-group using these matrices. Since the contact matrix and the infection rates vary over time, these fits are generated by providing a higher weight to the more recently seen data. Therefore, the method shows a better fit towards the end of the calibration period.

\begin{figure}
    \centering
    \includegraphics[width=\textwidth]{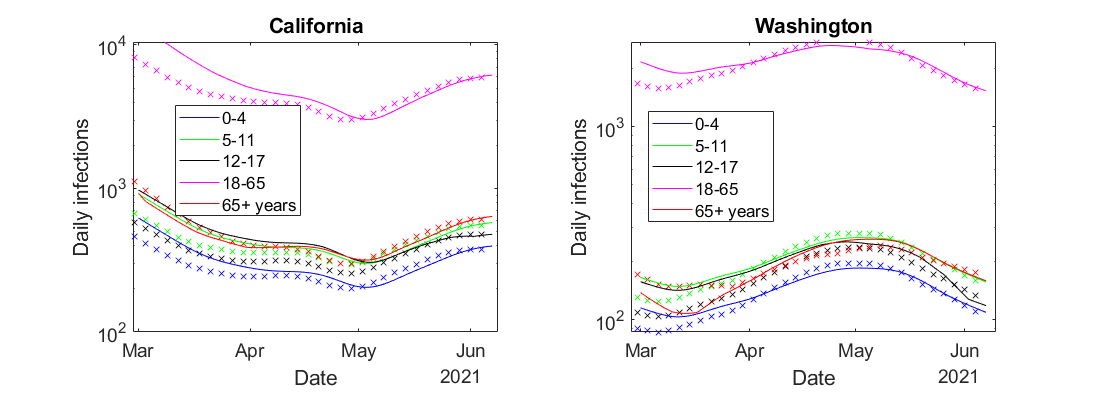}
    \caption{Result of fitting age-specific infection trends by learning contact matrix and infection rates. The markers indicate the pre-processed age-wise data and the solid lines indicate the fitted values. More weight is given to the end of the calibration period to identify recent rates.}
    \label{fig:age_fit}
\end{figure}
\subsection{Immune Escape Properties}
To incorporate immune escape variants like Omicron, a ``cross-protection'' matrix $C_P$ was introduced, such that $C_P(i, i')$ represents the protection against variant $i$ due to an infection by variant $i'$. This is the protection \textit{before} waning of immunity takes place. As a result, each variant has its own view of the population. In the absence of vaccines and waning immunity, susceptibility to variant $i$ will be given by:
\begin{equation}
    S^{(i)}(t) = \sum_{i'} \sum_{t'} (1-C_P(i, i')) I^{(i')}(t)
\end{equation}
The matrix $C_P$ is populated based on the scenario requirements. For instance, if Omicron is required to have an immune escape of 80\%, then
$C_{i, i'} = 1-0.8=2,$ where $i$ is a lineage of Omicron and $i'$ is a lineage of past variants. Otherwise, $C_{i, i'} = 1$.
Additionally, a vector $\mathbf{r}_{VE(j)}$ is introduced such that the $i^{th}$ element represents the initial protection against variant $i$ from the vaccine type $j$. In absence of any waning immunity and infection,
\begin{equation}
    S^{(i)}(t) = \sum_j (1-\mathbf{r}_{VE(j)}(i)) V_j(t)
\end{equation}
where $j$ is a vaccine type (first dose, 2nd dose, booster, ...) and $V_j(t)$ is the total individuals with vaccine type $j$.
Exact computation of susceptibility across variants in presence of waning immunity requires keeping track of various states and is discussed in Section~\ref{sec:all_state}.

\subsection{All-state Model}\label{sec:all_state}
\begin{figure}[!htbp]
    \centering
        \includegraphics[width=\textwidth]{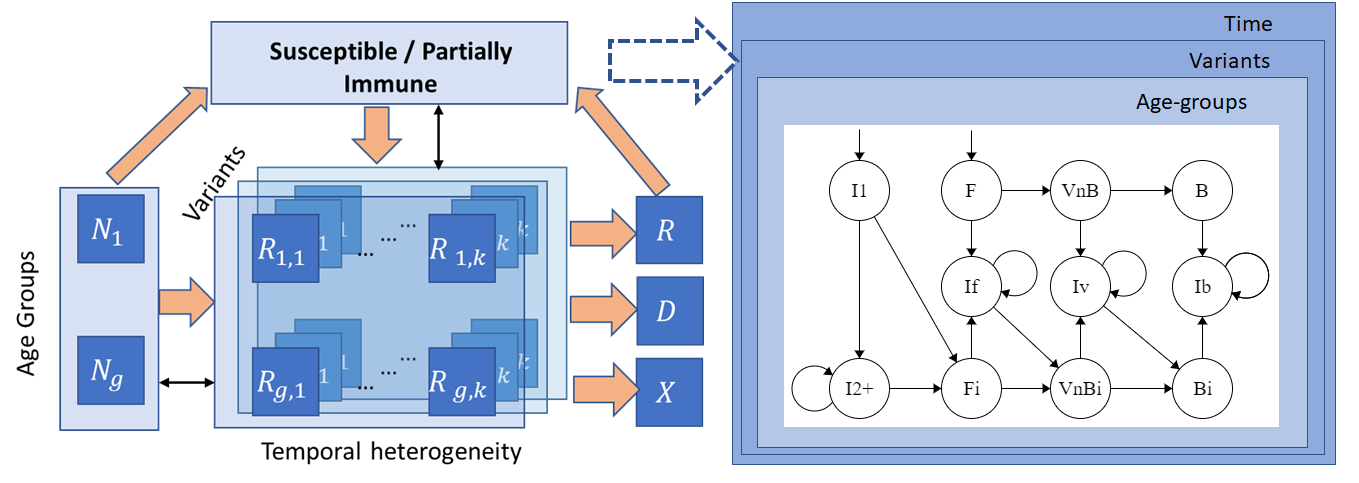}
  \caption{The ``all-state" version of the SIkJalpha model. The transmission is determined by those who were reported infected $1, \dots, k$ time-steps ago in groups $1,\dots, g$ (denoted by $R_{i, j}$) to create the new reported cases $R$. The values of $R_{i, j}$ also determine the new deaths $D$ and any other target ($X$) such as hospitalizations. The states on the right indicate various susceptibility states representing the combination of prior infection and vaccines, time of these events, variant that infected them, and their age group.}
  \label{fig:SIkJa}    
\end{figure}

To incorporate multiple reinfections, and various susceptibility groups that arise due to multiple rounds of vaccine, we define the following types of state: first-time infections ($I_1$), two or more infections ($I_2+$), first dose ($F$), infection with a prior first dose ($I_f$), first dose after a prior infection ($F_i$), second dose before boosting ($VnB$), infected with a prior second dose ($I_v$), second dose after a prior infection ($VnB_i$), booster dose ($B$), infected with a prior booster dose ($I_b$), and booster dose after a prior infection ($F_i$). Figure~\ref{fig:SIkJa} shows the allowed transitions between states. Let state $0$ denote those who are not in any of the above states (naive state). This is the initial state of everyone in the population before they get their first infection or vaccine.

The state of an individual is determined by the $4$-tuple $(j, i, g, t)$ determined by the state-type $j$, age group $g$, time $t$ and variant $i$.
Let $N(j, t, i, g)$ represent the number of individuals in state $(j, i, g, t)$. The association with the variant index represents the last variant that affected that individual. For instance, $N_{I2+}(g, i, t)$ would represent that the last transition they made on or before time $t$ was to a reinfection state due to variant $i$. Those in the vaccination-only state, i.e., $F$, $VnB$, and $B$ exist for all variants, i.e., $N_j(g, i, t) = N_j(g, i', t) \forall i, i'$ when $j \in \{F, VnB, B\}$.
There are at most two ways of transitioning from a state either due to an infection or due to vaccination. 

\textbf{Infection-induced transition. } At $t=0$, $N_j(g, i, t) = 0, \forall j,g,i,t$. Given how immunity wanes over time, we can compute the expected number of susceptible individuals in each state $(j,g, i', t')$ as perceived by each variant $i$ at time $t$.
\begin{align}
S_{i, t}(j, g, i', t') &= N(j, g, i', t'), j\in \{0\} \label{eqn:naive}\\
S_{i, t}(j, g, i', t') &= N(j, g, i', t')\bigg(1 - \mathbf{r}_{VE(j)}(i)\left(1-w_j(i, t-t')\right) \bigg)\,, j\in \{F, VnB, B\} \label{eqn:vax} \\
S_{i, t}(j, g, i', t') &= N(j, g, i', t')\bigg(1 - C_p(i, i')\left(1-w_j(i, t-t')\right) \bigg)\,, j \in \{I1, I2+\} \label{eqn:infec}\\
S_{i, t}(j, g, i', t') &= N(j, g, i', t')\bigg(1 - \max\{C_p(i, i'), \mathbf{r}_{VE(j)}(i)\}\left(1-w_j(i, t-t')\right) \bigg) \,, \nonumber\\
& j \in \{I_f, I_v, I_b, F_i, VnB_i, B_i\} \label{eqn:vax_infec}
\end{align}

In the naive state, everyone is susceptible (Equations~\ref{eqn:naive}). In a state-type with vaccination but no infection, probability of immunity is given by the initial vaccine efficacy and lack of waning immunity (Equations~\ref{eqn:vax}). In a state-type with prior infection but no vaccine probability of immunity is given by cross-protection and lack of waning immunity (Equations~\ref{eqn:infec}). In a state-type with prior infection and vaccination the probability of immunity is given by the maximum of cross-protection and the initial vaccine efficacy, if the protection has not waned (Equations~\ref{eqn:vax_infec}). 

Let $\Delta I^{(i)}(t, g)$ be the new infections at time $t$ in age group $g$ due to variant $i$. Then the new infections due to this variant in this age-group that will transition from state $(j, g, i, t)$ can be estimated as 
\begin{equation}
    \Delta_{I} N_t(j, g, i', t') = \Delta I^{(i)} (t, g) \frac{S_{i, t}(j, g, i', t')}{\sum_{j', t'}S_{i, t}(j', g, i', t')}
\end{equation}
where $\Delta_{I}$ represents the change due to infections at time $t$.

\textbf{Vaccination-induced transition. } 
In infection induced transitions we calculated susceptibility in each state to identify where the transitions are coming from. Similarly, in vaccine induced transitions we need to compute eligibility to identify where the transitions are coming from. The eligible population to go through vaccine-induced transition from the state $(j', g, i, t)$ due to vaccine type determined by $j$ (i.e., $F, VnB, B$), at time $t$ is given by:
\begin{align}\label{eqn:vax_induced_vax}
    E_{j,t}(j', g, i, t') &= N_t(j', g, i, t'), \mbox{ s.t., } j'\rightarrow j,\, t'< t-\delta(j' \rightarrow j), j\in\{F, VnB, B\}
\end{align}
where,  $j'\rightarrow j$ represents a valid transition, and $\delta(j' \rightarrow j)$ denotes the minimum delay allowed between the transitions. For instance, $\delta(j' \rightarrow j)$ may be set to six months for transitioning from 2nd dose only state ($VnB$) to the ``boosted without prior infections'' state $B$. Let 
$V_j(t, g)$ denote the number of vaccinations of type $j$ (first dose, second dose, or booster) given at time $t$ to the age-group $g$, then the
population in state $(j', g, i, t')$ transitioning out induced by vaccinating $V_j(t, g)$ individuals at time $t$ is
\begin{equation}
  \Delta_{V} N_t(j', g, i, t')   = V_j(t, g) \frac{E_{j,t}(j', g, i, t')}{\sum_{j', t'} E_{j,t}(j', g, i, t')}
\end{equation}

Finally, the change in number of individuals in each state is given by

\begin{align}
&\mbox{Initially, } \nonumber \\
&    N_0(0, g, i, t') = N(0, g) \forall g, i, t'\\
&    N_0(j, g, i, t') = 0\, \forall j>0, g, i, t' \\
&\mbox{Transitions to state $j$ at time $t$} \nonumber \\
&    N_t(j, g, i, t) = \sum_{i', t', j\rightarrow j'} \Delta_I N_t(j', g, i', t') + \sum_{i', t', j\rightarrow j'} \Delta_V N_t(j', g, i', t') \\
&\mbox{Transitions coming from state $j$ at time $t$} \nonumber \\
&    N_t(j, g, i, t') = N_{t-1}(j, g, i, t') - \sum_{i', t',, j\rightarrow j'} \Delta_I N_t(j', g, i', t') - \sum_{i', t',, j\rightarrow j'} \Delta_V N_t(j', g, i', t')
\end{align}

\begin{align}
    S_j^{(i)}(t, g) &= \sum_{i', t'} S_{i, t}(j, g, i', t')
\end{align}

As before, we assume that all severe outcomes at time $t$ in age-group $g$ ($X_j^{(i)}(t, g)$), including hospitalizations and deaths are simply linear functions of the temporal sequence of infections. 
\begin{equation}\label{eqn:target}
      \Delta X_j^{(i)}(t, g) = \sum_{l=1}^{k_X} \chi_l^{(i)}(g)(1 - P_j^{(i)}(t, g)) \left(I_j^{(i)}(t-(l-1)J, g) - I_j^{(i)}(t-lJ_X, g)\right).
\end{equation}
By this decoupling of the dependency, instead of sequential dependency between cases to hospitalizations to deaths, we are able to independently produce estimate parameters. Here $\chi_l^{(i)}(g)$ are rates of severe outcome for a ``naive'' individual - one who has no protection. We assume that the ratio of severity of any two variants is a constant. Therefore, we set $\chi_l^{(i)}(g) = \rho^{(i)}\chi_l(g)$, where $\rho^{(i)}$ represents relative severity of variant $i$. This parameter makes a significant difference only if there exist variants with significantly different severity in the last few days used for training ($~$100 days). For example, in the US, this parameter was used from December 2021 to March 2022, when the last 100 days were expected to have significant proportions of both omicron and delta variants.
We treat $\rho^{i}$ as a known parameter when provided in the scenarios~\cite{SMH}. Otherwise, we perform a grid search to identify the $\rho^{i}$ that best fits a held-out validation data (new deaths/hospitalizations) of the last week. When all the circulating variants are expected to have the same severity, we set $\rho^{(i)} = 1 \forall i$.

The parameter $P_j^{(i)}(t, g)$ is the protection against severe outcome of an individual in state-type $j$ in age group $g$ at time $t$ from variant $i$. We treat $P_j^{(i)}(t, g)$ as known parameters borrowed from the provided scenarios~\cite{SMH} for the scenario modeling hub. 
In absence of recommended values of protection (for forecast hubs), we use the various values suggested by the scenario modeling hub one at a time as ``sampled'' hyperparameters to generate ``sub-scenarios'' (see Section~\ref{sec:uncertainty}). It should be noted that by 2022, the naive population (no past vaccinations or infections) had become a small fraction of the population as suggested by the seroprevalence data~\cite{noauthor_nationwide_nodate}. Therefore, this parameter has a negligible impact on the projections. Since we do not observe severe outcomes per variant, we perform regression on:

\begin{equation}\label{eqn:target_final}
      \Delta X(t, g) = \sum_{l=1}^{k_X} \chi_l(g) \left(\sum_{i, j} \rho^{(i)}(1 - P_j^{(i)}(t, g)) \left(I_j^{(i)}(t-(l-1)J, g) - I_j^{(i)}(t-lJ_X, g)\right)\right).
\end{equation}
The quantity in the parenthesis is computed before the regression resulting in a linear regression only in $\chi_l(g)$, thus producing extremely fast results free from overfitting common in non-linear settings.

\subsection{Other Capabilities}
In addition to the above-described factors that are accounted for, the model also supports other capabilities as described below.

\textbf{Anomaly detection and Smoothing. } To address the reporting noise and delayed dumps of backlogged cases and deaths, various strategies were used. Until July 2020, only moving average smoothing was used. However, delayed dumps caused large spikes, thus misleading the model. Therefore, an anomaly detection module was added which considered a daily value (case or death) to be an outlier if it exceeded 4 times the median value of the last few ($\sim$70) days. The outlier is then replaced with the linearly interpolated value from the neighboring non-outliers. Then we apply a moving average smoothing.
With changes to reporting schedule (once a week, reporting only on certain weekdays, etc.) the anomaly detection module had to be changed. We first accumulate the reported targets into a weekly time series instead of daily. The anomaly detection is then applied to this weekly time series. Then we convert this time series back into daily time series by distributing the reports in a week equally among its days. Finally, we apply to smooth on this time series.

\textbf{Future Contact Behavior. } Future changes in Non-Pharmaceutical Interventions (NPIs) for scenario modeling are implemented as linear scaling of contact matrix~\ref{sec:contact}. In absence of age stratification, this results in a linear scaling of the transmission rates in the future. In earlier rounds, when NPIs were asked to be lifted linearly, the scaling was changed linearly over time. In rounds where future NPI modeling was left up to the teams, the scaling was based on contact rates of the previous year obtained from Cuebiq~\cite{cuebiq} to account for seasonal contact changes.

\textbf{Arbitrary variant. } Arbitrary variants in the future can be incorporated into the model along with the existing ones. The model requires as input the day of introduction (or distribution over days), the number of introductions, and properties of the new variant (transmissibility and immune escape advantages).

\textbf{More vaccine doses. } While we present the model with three doses (first, second, and booster), the model supports an arbitrary number of doses. This is done by repeating Equations~\ref{eqn:vax}, and~\ref{eqn:vax_induced_vax} to more vaccine doses beyond $\{F, VnB, B\}$.
\begin{table}[!ht]
    \centering
    \begin{tabular}{|p{1.8in}|p{3in}|}
    \hline
         Learnable Parameters (estimated by regression) & Transmission rates vector per variant $i$ $\beta_l^{(i)}$, contact matrix $C$, severe outcome rates $\chi_l(g)$ for deaths and hospitalizations \\
         \hline
         Selected Hyperparameters (selected by grid search) & For cases $k=2, J=7, \alpha \in \{0.9\}$, for deaths $k_D \in \{3,4,\dots,7\}, J_D=7, \alpha_D \in \{0.95, 0.98\}$, for hospitalizations $k_H = \{2, 3\}, J_H = \{2, 3, 7\}, \alpha_H \in \{0.95, 0.98\}$, relative variant severity $\rho^{(i)}$\\
         \hline
         Provided/Sampled Hyperparameters (sampled from a reasonable set) & under-reporting/true cases interval, interval estimate of variant proportions (affecting $\Delta I^{(i)}(t)$), waning immunity parameters $\epsilon_A(i, j), w_j(i, \tau)$, protection against severe disease $P^{(i)}_j(t,g)$, Relative vaccine efficacy per variant per group $\mathbf{r}_{VE(j)}(i)$, cross protection matrix $C_P$, future vaccine coverage $u$\\
         \hline  
    \end{tabular}
    \caption{Categorization of parameters involved in model training.}
    \label{tab:params}
\end{table}

\section{Model Uncertainty}
\label{sec:uncertainty}
Two methods have been used to model uncertainty: sub-scenarios and Random Forest. The latter has only been used for the Germany/Poland Forecast Hub.

\subsection{Sub-scenario Based}
Here, our goal is to generate multiple future trajectories from which we can sample desired quantiles. We divide the parameters involved in the model into three categories (Table~\ref{tab:params}): (i) Learnable parameters include those that are estimated using regression. (ii) Selected hyperparameters are those that are identified via grid search based on a held-out validation set (combination of hyperparameters that best predicts the unseen data), and (iii) Provided/Sampled hyperparameters are those that are either provided by the scenarios or are considered to take one of several different values. Among the provided/sampled hyperparameters that can take several values, each setting corresponds to a ``sub-scenario''. Once we fix a sub-scenario, learnable parameters can be estimated and selected hyperparameters can be identified by grid-search, leading to one trajectory. We form all possible combinations of sub-scenarios thus typically getting 50-300 trajectories. The desired quantiles are then sampled from these trajectories.

\subsection{Random Forest}
In this approach, we first generate the ``mean'' prediction. This is the time-series outcome of regression which we convert to a vector of four elements $(y_1, y_2, y_3, y_4)$. The $i^{th}$ element represents the $i$-week ahead forecast. We assume that the true observation $i$-week ahead is $\bar{y}_i = y_i + e_i$. For each $i$, to estimate $e_i$, we train a Random Forest~\cite{meinshausen_quantile_2006} with 100 trees. The inputs to the Random Forest are current and last week's incident data (cases or deaths), the population of the region, and the prediction from SIkJalpha, $y_i$. The output is the deviation from the true observed data $\bar{y}_i - y_i$. The training is performed by considering recently seen (last 35 days) data, and the difference between prediction and observed data over all the regions. The desired quantiles are then sampled from the Random Forest.

\section{Real-time Prediction}
\label{sec:real-time}
\subsection{Participation in Real-time Prediction Efforts}
Various versions of SIkJalpha model were submitted to US/CDC Forecast Hub (US FH), Europe/ECDC Forecast Hub (Euro FH), Germany/Poland Forecast Hub (G-P FH), US/CDC Scenario Modeling Hub (US SMH), and Europe/ECDC Scenario Modeling Hub (EU SMH). Submissions to the forecast hubs were made every week on Sunday. Scenario Modeling Hubs have submissions once in 1-2 months. The details of targets and modeling choices are presented in Table~\ref{tab:all_models}.

\subsection{Real-time Publicly Available Forecasts}
We have maintained an online dashboard since April 2020~\cite{recover_github} to visualize the generated forecasts for US states and countries around the world. These forecasts are generated using the same modeling choices as mentioned in Table~\ref{tab:all_models} for US FH and Euro FH, respectively.
Since December 2020, we have also added forecasts for all the locations around the world~\cite{google_forecasts} for which Google provides case and death data~\cite{noauthor_covid-19_nodate-2}. This includes 214 Admin-0 locations, around 1000 Admin-1 locations, and around 16,000 Admin-2 locations. This implementation includes modeling choices identical to G-P FH with the exception of the inclusion of 2-dose vaccines which are assumed to provide all-or-nothing protection. When vaccination data is not available from Google for Admin-1 or Admin-2 locations, we assume that the vaccination uptake per unit population is the same as that of the higher Admin-level for which data is available. We do not generate uncertainty bounds for these forecasts.

\subsection{Implementation}
All code was written in MATLAB and is publicly available on Github~\cite{recover_github}. Daily forecasts have been generated through an automated script at 5 a.m. Pacific Time every day since May 2020 on an Intel 2-core \textbf{desktop}. The scenario projections were developed and generated on the same machine as per the submission schedule until September, after which we moved to a Dell PowerEdge R540 \textbf{server}. It has two Intel Xeon Gold 5218 processors each running @2.3G, 16C/32T, 10.4GT/s, 22M Cache (total 32 core/64 threads, 44M Cache). It has 64GB RDIMM, 3200MT/s, Dual Rank RAM, and 1.92TB SSD SATA Read Intensive 6Gbps drive. MATLAB implicitly converts many matrix operations and functions into multi-threading. Starting in September, we also introduced parallelism across sub-scenarios. The runtimes of the latest implementation are presented in Table~\ref{tab:runtimes}. The most time-consuming aspect is the estimation of the susceptible population which requires evaluating all states from the beginning of the epidemic, which in our model is set to Jan 23, 2020. To accelerate the computation, we update the states at a weekly granularity rather than daily.

\begin{table}[!ht]
    \centering
    \begin{tabular}{l|c|c}
    Targets & ``all-state''? & Runtime/subscenario \\
    \hline
        US states / US FH & Yes &  25s on server\\       
        $\sim$190 countries & Yes & $\sim$200s on desktop \\
        Euro FH/SMH & Yes & 15s on server \\
        GP FH & No & $<$ 10s on desktop (incl. quantiles)\\
        $>$17,000 locations (Google) & No & $\sim$400s on desktop
    \end{tabular}
    \caption{Approximate runtimes of various implementations. }
    \label{tab:runtimes}
\end{table}
\section{Evaluations}
\label{sec:eval}

\begin{figure}[!ht]
    \centering
    \begin{subfigure}{0.49\textwidth}
    \includegraphics[width=\textwidth]{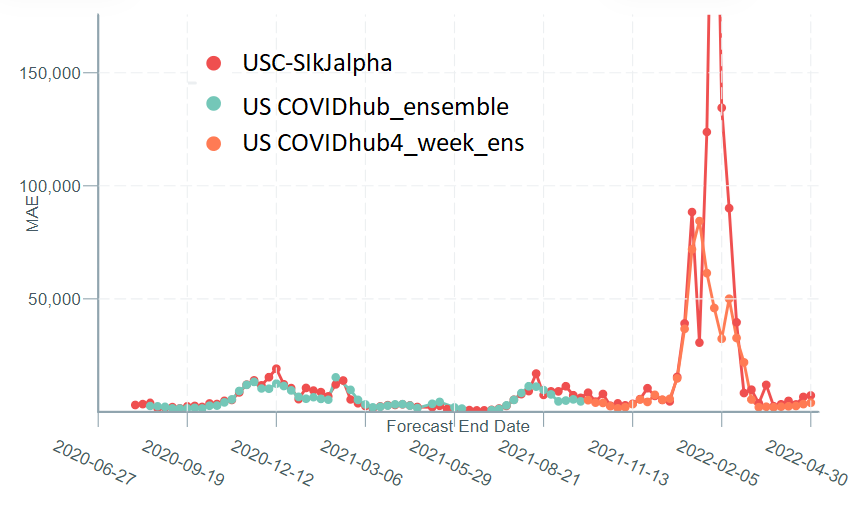}    
    \caption{US Forecast Hub - cases}
    \end{subfigure}
    \begin{subfigure}{0.49\textwidth}
    \includegraphics[width=\textwidth]{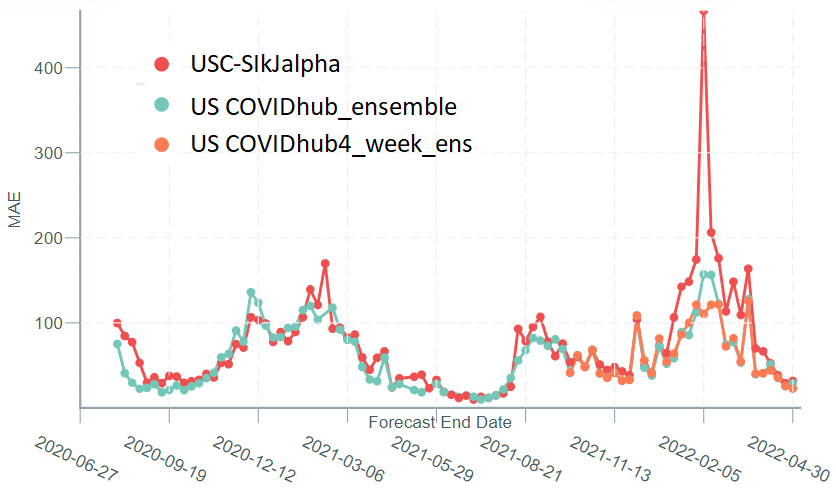}    
    \caption{US Forecast Hub - deaths}
    \end{subfigure}
    \begin{subfigure}{0.49\textwidth}
    \includegraphics[width=\textwidth]{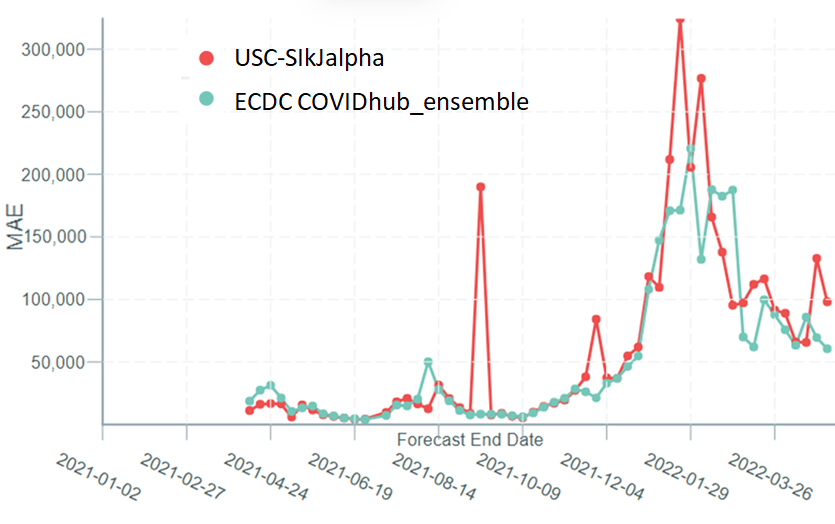}    
    \caption{ECDC Forecast Hub - cases}
    \end{subfigure}
    \begin{subfigure}{0.49\textwidth}
    \includegraphics[width=\textwidth]{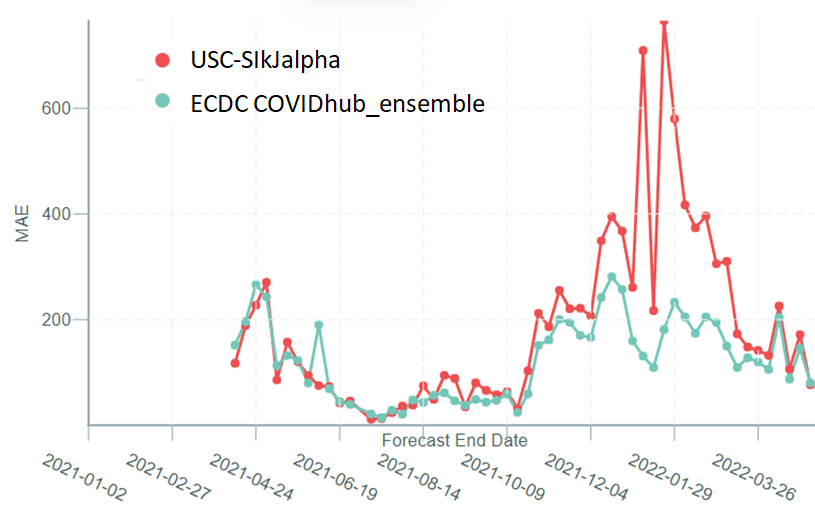}    
    \caption{ECDC Forecast Hub - deaths}
    \end{subfigure}
    \caption{Comparison of our submitted 4-week against Hub ensemble in terms of MAE}
    \label{fig:MAE}
\end{figure}

\begin{figure}[!ht]
    \centering
    \begin{subfigure}{0.49\textwidth}
    \includegraphics[width=\textwidth]{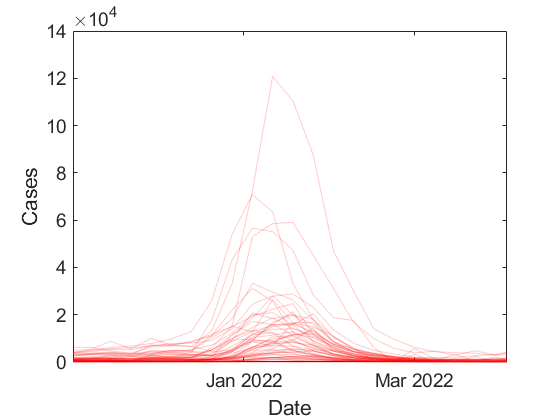}    
    \caption{Cases in the US states}
    \end{subfigure}
    \begin{subfigure}{0.49\textwidth}
    \includegraphics[width=\textwidth]{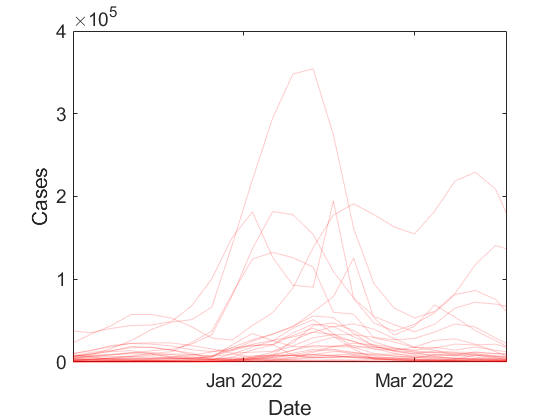}    
    \caption{Cases in the European countries}
    \end{subfigure}
    \caption{Reported cases after applying the same smoothing to both datasets: States in the US have smoother trends than the European countries.}
    \label{fig:omicorn-trends}
\end{figure}

\begin{figure}[!ht]
    \centering
    \includegraphics[width=0.7\textwidth]{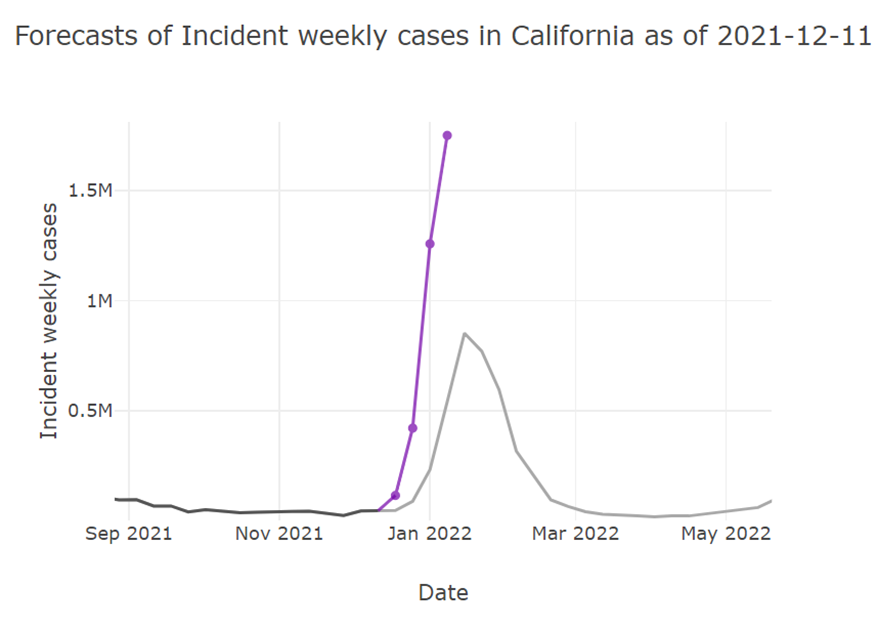}
    \caption{The model often over-predicts just before the surge: Here the purple line is the model prediction for California before the Omicron surge. The grey line is the ground truth. Due to the inclusion of variants-specific data, the model can predict an impending surge, however, it over-predicts the intensity of the surge.}
    \label{fig:overpredicts}
\end{figure}

Our models have been evaluated by us~\cite{srivastava_covid19_2020} as well as several other studies~\cite{bracher_assembling_nodate,cramer2022evaluation,bracher2021pre,friedman2021predictive,howerton2023informing}, reports~\cite{US_eval_report,EU_eval_report,KITeval} and dashboards~\cite{US_eval_dashboard}. Despite the lack of human involvement in generating forecasts (except for debugging and adding new features), our methods have been in the top 25\% during most of the pandemic based on mean absolute error and weighted interval score~\cite{bracher2021evaluating}. For case forecasts in the US, among all models that were submitted consistently in 2020-2021, our submission ranks highest in terms of Weighted Interval Score~~\cite{lopez2023challenges}. For country-level death forecasts over multiple countries, our forecasts performed the best in terms of median absolute percentage error~\cite{friedman2021predictive} among 7 models that fit the inclusion criteria among 386 models analyzed. This analysis included long-term forecasts up to 12 weeks ahead generated in 2020 to mid 2021. 

However, the US hospitalization forecasts have not been good. This is particularly because it was assumed that there is at least one one-week lag between the reported cases and hospitalizations. However, this was not consistently true and in the majority of the weeks, the reported cases were well-aligned (without a lag) with reported cases due to delay in reporting. It was fixed in February of 2022, by removing the enforced one-week lag. This is a drawback of the ``automated'' approach and could have been fixed earlier with more human involvement -- by visualizing the results every week. However, due to the involvement of the lone author in multiple collaborative efforts, this was not feasible and would have required multiple personnel. 

Figure~\ref{fig:MAE} shows the comparison of our submissions and the Hub ensemble in terms of mean absolute error for US and ECDC 4-week ahead case and death forecasts.We observe that until the end of 2021, the performance was close to the Hub ensembles. The case forecasts had poor performance just before and during the rise of the Omicron variant. The errors are higher in the US compared to those in the ECDC. This could be due to sparser state-level data on variants leading to poor estimates of variant proportions. Death forecasts are impacted less in the US, with the exception of one outlier. On the other hand, ECDC death forecasts are impacted more during this period. This may be due to the poorer quality of case reporting during this period in Europe compared to the US states (see Figure~\ref{fig:omicorn-trends}). Since, during this time, our death forecasts were based on a function of reported cases, death forecasts were of relatively poor quality as well.

Our models often correctly predict the existence of an upcoming surge, but they tend to over-predict its intensity (Figure~\ref{fig:overpredicts}). While qualitatively, predicting the existence of a surge is useful, with traditional metrics such as mean absolute error, the models are penalized for overpredicting. How to improve the forecast based on these metrics during a time of variant-drive surge is not clear. Better quality of data on genomic surveillance and case reporting can help with better estimation of variant proportions, resulting in better forecasts, at least in the short term.

\section{Conclusions}

We presented various versions of the SIkJalpha model that have been used over time over different forecast hubs and scenario modeling hubs. The models have evolved into one capable of accounting for various age groups, an arbitrary number of variants, imperfect protection from natural and vaccine-induced immunity, and immunity states formed by multiple vaccine doses and infections.
We have presented the implementation details including data sources, parameter estimation methods, runtimes, and uncertainty quantification. The goal of this paper was to provide an overview of the modeling choices over time. The details of the impact of each modeling decision are beyond the scope of this paper. Instead, we focused on the breadth of modeling decisions. Our forecasts for the forecast hubs were generated through automated scripts. The advantage of our approach is scalability due to automation and fast implementation. That is, we do not perform any manual tuning or analyze individual forecasts for each region. The same advantage has the drawback that the models may produce poor forecasts if there is an anomaly or high noise in data that have not been already accounted for in the scripts. A particular period where the models perform poorly is just before a surge due to a new variant. The model is able to predict the existence of a surge, however, it often over-predicts the intensity. 

In future work, we will perform retrospective analysis to identify which modeling choices lead to better performance. Particularly, how to predict the existence of a surge without over-predicting its intensity. This will include identification of the transmissibility advantage, the immune escape advantage, and better estimates of population immunity. Estimating population immunity will require exact identification of the dynamics of waning immunity, which so far has been borrowed from the suggestions of the Scenario Modeling Hub. 

\section*{Acknowledgments}
This work was supported by the  Centers for Disease Control and Prevention and the National Science Foundation under the awards no. 2027007, 2135784, and 2223933. Any opinions, findings, and conclusions or recommendations expressed in this material are those of the author and do not necessarily reflect the views of the National Science Foundation and the Center for Disease Control and Prevention.
The author would like to thank Tianjuan Xu and Majd Al Aawar for assistance with the submission of forecasts to the forecast hubs. The author would also like to thank the US Scenario Modeling Hub for useful discussions that have helped with the evolution of the models.

\begin{landscape}
\begin{table}[p]
\caption{Modeling choices for various Hub submissions.}
\label{tab:all_models}
\resizebox{1.6\textwidth}{!}{%
\begin{tabular}{|l|c|c|c|c|c|} 
\toprule
& \begin{tabular}[c]{@{}c@{}}US FH\\ (start: 06/20)\end{tabular}  & \begin{tabular}[c]{@{}c@{}}G+P FH\\ (start: 11/20)\end{tabular} & \begin{tabular}[c]{@{}c@{}}Euro FH\\ (start: 02/21)\end{tabular}  & \begin{tabular}[c]{@{}c@{}}US SMH\\ (start: 01/21)\end{tabular}  & \begin{tabular}[c]{@{}c@{}}Euro SMH\\ (start: 03/22)\end{tabular}  \\ 
\midrule
Cases  & Yes  & Yes  & Yes  & Yes  & Yes  \\ 
\midrule
Deaths  & Yes  & Yes  & Yes  & Yes  & Yes  \\ 
\midrule
Hospitalizations  & 11/20-2/23  & -  & -  & Yes  & -  \\ 
\cmidrule{1-4}\cline{5-5}\cmidrule{6-6}
True Infections  & \begin{tabular}[c]{@{}c@{}}Estimated: start-12/20\\ Seroprevalence: 1/21-3/22\\ Wastewater: 3/22-2/23\end{tabular} & Extrapolated  & Extrapolated  & \begin{tabular}[c]{@{}c@{}}Seroprevalence: start-3/22\\ Wastewater: 3/22-2/23\end{tabular}  & Extrapolated  \\ 
\midrule
Vaccine  Adoption  & \begin{tabular}[c]{@{}c@{}}Linear: 1/21 – 7/21\\Contagious: 7/21 – 2/23\end{tabular}  & -  & \begin{tabular}[c]{@{}c@{}}Linear: start – 7/21\\Contagious: 7/21 – 2/23\end{tabular}  & \begin{tabular}[c]{@{}c@{}}Sigmoid: start – 7/21\\Contagious: 7/21 – 2/23\end{tabular}  & \begin{tabular}[c]{@{}c@{}}Contagious:\\ start – 2/23\end{tabular}  \\ 
\midrule
\multirow{2}{*}{Vaccine  Dosage}  & \multirow{2}{*}{\begin{tabular}[c]{@{}c@{}}2-dose: 1/21 – 7/21\\All-doses: 07/21 – 2/23\end{tabular}}  & \multirow{2}{*}{-}  & \multirow{2}{*}{\begin{tabular}[c]{@{}c@{}}2-dose: start – 7/21\\2-dose+ boost: 07/21 – 2/23\end{tabular}} & 2-dose: start – 7/21  & \multirow{2}{*}{\begin{tabular}[c]{@{}c@{}}2-dose+booster:\\ start - 2/23\end{tabular}} \\ 
\cmidrule(r){5-5}
&  &  &  & 2-dose+boost: 07/21 – 2/23  &  \\ 
\cmidrule{1-1}\cmidrule(lr){2-2}\cmidrule{3-3}\cmidrule(l){4-4}\cmidrule{5-6}
Vaccine  Model  & \begin{tabular}[c]{@{}c@{}}AoN: 1/21 – 7/21\\Waning: 7/21 – 2/23\end{tabular}  & -  & \begin{tabular}[c]{@{}c@{}}AoN: start – 7/21\\Waning: 7/21 – 2/23\end{tabular}  & \begin{tabular}[c]{@{}c@{}}AoN: start – 7/21\\Waning: 7/21 – 2/23\end{tabular}  & -  \\ 
\cmidrule{1-1}\cmidrule(lr){2-2}\cmidrule{3-3}\cmidrule(lr){4-5}\cmidrule{6-6}
\begin{tabular}[c]{@{}l@{}}Natural Immunity\\ Waning\end{tabular} & 7/21 - 2/23  & -  & 7/21 - 2/23  & 7/21  – 2/23  & start - 2/23  \\ 
\midrule
Variants  & \begin{tabular}[c]{@{}c@{}}2-strain: 03/21 - 07/21\\All PANGO: 07/21 - 12/21\\Selected: 12/21 - 2/23\end{tabular} & -  & \begin{tabular}[c]{@{}c@{}}All PANGO: 07/21 - 12/21\\Selected: 12/21 - 2/23\end{tabular}  & \begin{tabular}[c]{@{}c@{}}2-strain: 03/21 - 07/21\\All PANGO: 07/21 - 12/21\\Selected: 12/21 - 2/23\end{tabular} & start - 2/23  \\ 
\midrule
Immune Escape  & Since  12/21  & -  & Since  12/21  & Since  12/21  & Yes  \\ 
\midrule
Age groups  & -  & -  & -  & Yes  & -  \\ 
\midrule
\begin{tabular}[c]{@{}l@{}} All Susceptibility\\ States\end{tabular}  & Since  12/21  & -  & Since  12/21  & Since  12/21  & Yes  \\ 
\midrule
Quantiles  & \begin{tabular}[c]{@{}c@{}}Sub-scenarios:\\ 11/20-2/23\end{tabular}  & Random Forest  & \begin{tabular}[c]{@{}c@{}}Sub-scenarios:\\ 11/20-2/23\end{tabular}  & \begin{tabular}[c]{@{}c@{}}Sub-scenarios:\\ 1/21-2/23\end{tabular}  & \begin{tabular}[c]{@{}c@{}}Sub-scenarios:\\ start-2/23\end{tabular}  \\
\bottomrule
\end{tabular}
}
\end{table}
\end{landscape}

\bibliographystyle{plain}
\bibliography{refs,epidemic,my_pubs}

\end{document}